\documentclass[aps,pra,showpacs,showkeys, superscriptaddress, twocolumn]{revtex4-1}
\usepackage{xr}
\usepackage{float}
\usepackage[toc,page]{appendix}
\usepackage[utf8x]{inputenc}
\usepackage[bottom]{footmisc}
\usepackage{amsmath,amsthm,amssymb}
\usepackage{mathrsfs}
\usepackage{bbold}
\usepackage{array}
\usepackage{dsfont}
\usepackage{bbm}
\usepackage{graphicx}
\usepackage{calc}
\makeatletter
\def\be{\begin{equation}}
\def\ee{\end{equation}}
\def\bea{\begin{eqnarray}}
\def\eea{\end{eqnarray}}
\def\ben{\begin{equation*}}
\def\een{\end{equation*}}
\def\bean{\begin{eqnarray*}}
\def\eean{\end{eqnarray*}}
\def\bma{\begin{mathletters}}
\def\ema{\end{mathletters}}
\def\bi{\begin{itemize}}
\def\ei{\end{itemize}}
\newtheorem{defn}{Definition}
\newtheorem{thm}{Theorem}
\newtheorem{cor}[thm]{Corollary}
\newtheorem{lem}{Lemma}
\newtheorem{example}{Example}

\newcommand{\pk}[3]{\langle #1 \,  | \, #2 \, | #3  \rangle}
\newcommand{\ket}[1]{ | \, #1  \rangle}
\newcommand{\bra}[1]{ \langle #1 \,  |}
\newcommand{\proj}[1]{\ket{#1}\bra{#1}}

\newcommand{\ketbra}[2]{|{#1}\rangle \langle {#2}|}

\newcommand{\braket}[2]{\langle #1 | #2 \rangle}

\newcommand{\C}[1]{\mathbb{C}^{#1}\otimes\mathbb{C}^{#1}}
\newcommand{\s}[1]{\scriptscriptstyle #1}
\begin{document}

\title{Necessary Condition for Local Distinguishability of Maximally Entangled States: Beyond Orthogonality Preservation}
\author{Tanmay Singal}
\email{tanmaysingal@gmail.com}
\affiliation{Department of Applied Mathematics, Hanyang University (ERICA),
55 Hanyangdaehak-ro, Ansan, Gyeonggi-do 426-791, Korea}
\author{Ramij Rahaman}
\email{ramijrahaman@gmail.com}
\affiliation{Department of Mathematics, University of Allahabad, Allahabad 211002, U.P., India}
\author{Sibasish Ghosh}
\email{sibasish@imsc.res.in}
\affiliation{Optics \& Quantum Information Group, The Institute of Mathematical Sciences, HBNI, CIT Campus, Taramani, Chennai, 600 113, India}
\author{Guruprasad Kar}
\email{gkar@isical.ac.in}
\affiliation{Physics \& Applied Mathematics Unit, Indian Statistical Institute, 203 B. T. Road, Kolkata - 700108, India}

\begin{abstract}
The (im)possibility of local distinguishability of orthogonal multipartite quantum states still remains an intriguing question. Beyond $\C{3}$, the problem remains unsolved even for maximally entangled states (MES). So far, the only known condition for the local distinguishability of states is the well-known orthogonality preservation (OP). Using an upper bound on the locally accessible information for bipartite states, we derive a very simple necessary condition for any set of pairwise orthogonal MES in $\C{d}$ to be perfectly locally distinguishable. This condition is seen to be stronger than the OP condition. This is particularly so for any set of $d$ number of pairwise orthogonal MES in $\C{d}$. When testing this condition for the local distinguishability of all sets of four generalized Bell states in $\C{4}$, we find that it is not only necessary but also sufficient to determine their local distinguishability. This demonstrates that the aforementioned upper-bound may play a significant role in the general scenario of local distinguishability of bipartite states.
\end{abstract}



\maketitle

\section{Introduction}

Various quantum information processing tasks are often restricted by local operations and classical communication (LOCC) protocols. Such LOCC protocols are natural requirements when two or more physically separated parties have to accomplish a quantum information processing task. Local discrimination of pairwise orthogonal quantum states is one such task. This topic in quantum information theory has received considerable attention in recent years to understand more deeply the role of entanglement and non-locality in quantum information processing. \medbreak

When restricted to LOCC, two distant parties Alice and Bob generally cannot distinguish among bipartite states as efficiently as they can with global operations. This is also true when the states are pairwise orthogonal. As is well-known, pairwise orthogonal states can be perfectly distinguished by global operations. Such a projective measurement is generally a global operation, and cannot be achieved through LOCC protocols. Hence generally, pairwise orthogonal bipartite states aren't perfectly distinguishable by LOCC.
\medbreak

Walgate et al. \cite{W00} showed that by using LOCC only, any two pure orthogonal multipartite states can be perfectly distinguished. The local indistinguishability of pairwise orthogonal multipartite states is a signature of the non-locality expressed by these states. Since entanglement is deeply connected with non-locality one would be well placed to assume that pairwise orthogonal product states would be perfectly distinguishable by LOCC. This assumption, however, is wrong, for instance, Bennett et al. \cite{BD99} showed that there exist a set of nine unentangled pure orthogonal states in $\C{3}$ which cannot be perfectly distinguished by LOCC. Therefore, the local distinguishability problem is not so simple, particularly when the number of states is more than two.

Fan \cite{F04} showed that when $d$ is a prime number and $l$ is a positive integer such that $l(l-1) < 2 d$, then any $l$ number of mutually orthogonal generalized Bell states in $\C{d}$ (described in the teleportation paper of Bennet et al \cite{B93} and defined in definition \eqref{def1} of the present paper) are perfectly distinguishable by LOCC itself. The question of perfect local discrimination of pairwise orthogonal generalized Bell states in $\C{d}$ was later raised by Ghosh et al. \cite{G04} for \emph{general} $d$. They showed that if $k > d$, \emph{no} set of $k$ number of generalized Bell states in $\C{d}$ can be perfectly distinguished by LOCC. In the context of \emph{general} maximally entangled states (MES) in $\C{d}$, it is known that no set of $k>d$ pairwise orthogonal MES in $\C{d}$ are perfectly distinguishable by LOCC \cite{N05}. Moreoever, as a general result, it has been shown in \cite{N05} that \emph{any three} pairwise orthogonal MES in $\C{3}$ are perfectly distinguishable by LOCC. Yu et al. \cite{Y12} provided the first example of a set of four pairwise orthogonal MES in $\C{4}$ (of the form $\frac{1}{2}\sum_{i,j=0}^{1} \ket{ij}\otimes (\sigma_\alpha \ket{i} \otimes \sigma_\beta \ket{j})$, where $\sigma_\alpha,\sigma_\beta \in \{ \mathbb{1}_2, \sigma_x, \sigma_y, \sigma_z\}$) which are not perfectly distinguisable by LOCC. In fact, they showed that these four states are not distinguisable when the POVMs are restricted to PPT-POVM type, i.e., each effect of the POVM is PPT. PPT-POVMs are more general than LOCC.  Imperfect distinguishability of these states was later discussed by Cosentino \cite{C13}. In \cite{N13}, Nathanson gives an example of \emph{three} pairwise orthogonal MES in $\C{4}$ which are not perfectly distinguishable by one-way LOCC, but are so by two-way LOCC.

In most of the aforementioned results, the sets of states exhibited a certain symmetry which was used to conclude whether the states were locally distinguishable or not. For general sets of orthogonal bipartite states, the only known condition for their local distinguishability is the \emph{existence} of  \emph{orthogonality preserving} (OP) measurements on Alice's and Bob's sides. A measurement is OP if, for any measurement outcome, the corresponding post-measurement selected states remain pairwise orthogonal. The existence of OP measurements is a necessary (but not sufficient) condition for local distinguishability because for the post-measurement selected states to remain distinguishable it is required that they be pairwise orthogonal.

Locally accessible information of a set of states is the amount of classical information one can retrieve from the states using LOCC. For any $m$ quantum states to be locally distinguishable, it is necessary that the locally accessible information of these $m$ states should be at least $log_2 m$ bits. When performing a measurement on the states, it is possible that the locally accessible information of the post-measurement selected states decreases below $log_2 m$ bits. In such a case, the states become locally indistinguishable. Thus, the condition to maintain the locally accessible information to at least $log_2m$ bits, is a constraint on the measurements, just like OP is a constraint on the measurements. In fact this condition is at least as strong as the OP condition because it subsumes the OP condition.

In this paper we use the Holevo-like upper bound on locally accessible information \cite{B03} to give a necessary condition for the local distinguishability of MES. In the LOCC protocol for the local distinguishability of MES, when this upper bound (for the post-measurement selected states) becomes smaller than $log_2 m$ bits, then said states are locally indistinguishable. To show that our necessary condition is indeed stronger than OP, we show that there exist some sets of four generalized Bell states in $\C{4}$, for which the OP condition isn't strong enough to conclusively determined the states' local indistinguishability whereas our condition is strong enough to do so. We also analyze the local (in)distinguishability of all sets of four generalized Bell states in $\C{4}$. For this, we first partition $\binom{16}{4}$ such distinct sets into $122$ equivalence classes, where all sets within an equivalence class are related to each other by a local unitary transformation (see Appendix \ref{122} for the details). We represent each equivalence class by a constituent set of (generalized Bell) states. We then use our necessary condition to isolate those sets which fail the necessary condition test which means that these sets of states are locally indistinguishable. Of the total $122$, such sets are $39$ in number. Surprisingly, we find that states in each of the remaining $83$ sets are distinguishable by one-way LOCC. This shows that our condition is not only necessary but also sufficient to determine the local distinguishability of four generalized Bell states in $\C{4}$. This signifies the role of the Holevo-like upper upper for locally accessible information in the context of local distinguishability of quantum states.

The paper is divided into the following sections: in section \ref{secneccond} we obtain the necessary condition for local distinguishability, in section \ref{secensgen} we provide an example of how to use the necessary condition. This also demonstrates that the OP condition is weaker than the necessary condition we derived in section \ref{secneccond}. Also, after the example, we apply the necessary condition for the local distinguishability of \emph{all} sets of four generalized Bell states in $\C{4}$. Section \ref{conclusion} concludes the paper with future directions. In Appendix \ref{122}, describe how to find out the $122$ equivalence classes of sets of four generalized Bell states in $\C{4}$. In Appendix \ref{app} we list the LOCC protocols used to distinguish sets of four generalized Bell states (in $\C{4}$), which satisfy the necessary condition. 

\section{Necessary Condition for Local Distinguishability of Ensemble of Maximally Entangled States}
\label{secneccond}

Consider a set of $m$ pairwise orthogonal MES $ \ket{\psi_1}$, $\ket{\psi_2}$, $\cdots$, $\ket{\psi_m}$ $\in$ $\C{d}$. Let Alice control one subsystem and Bob the other. Let $\rho^{(A)}_i = Tr_B(\ketbra{\psi_i}{\psi_i})$ and $\rho^{(B)}_i = Tr_A(\ketbra{\psi_i}{\psi_i})$ be the $i$-th reduced state on Alice's subsystem and Bob's subsystem respectively. Let Alice start the LOCC protocol with some measurement, whose Kraus operators are $\{ K_i \}_{i=1}^{n}$, where $n \ge 2$ and $\sum_{i=1}^{n} K_i^\dag K_i= \mathbb{1}_A$, where $\mathbb{1}_A$ is the identity operator acting on Alice's subsystem. Let the measurement yield the $\alpha$-th outcome. Thus the post-measurement state is given by
\begin{equation}
\label{PMS}
\ket{\psi_{i}} \longrightarrow \ket{\psi_{\s{i,\alpha}}} = \frac{K_\alpha\otimes \mathbb{1}_B}{\sqrt{\pk{\psi_i}{K_\alpha^\dag K_\alpha \otimes \mathbb{1}_B}{\psi_i}}} \ket{\psi_i},
\end{equation}
for all $1 \leq i \leq m$, and where $\mathbb{1}_B$ is the identity operator acting on Bob's subsystem. Let $\rho^{(A)}_{\s{i,\alpha}} \equiv Tr_B \left( \proj{\psi_{i,\alpha}}\right)$, and $\rho^{(B)}_{\s{i,\alpha}}\equiv Tr_A \left( \proj{\psi_{i,\alpha}} \right)$ be the post measurement reduced states (PMRS) on Alice's and Bob's sides respectively. Then the average PMRS on Alice's and Bob's sides are $\rho^{(A)}_{\s{\alpha}} = \sum_{i=1}^{m} \frac{1}{m} \rho^{(A)}_{\s{i,\alpha}}$ and $\rho^{(B)}_{\s{\alpha}} = \sum_{i=1}^{m} \frac{1}{m} \rho^{(B)}_{\s{i,\alpha}}$, respectively, where the $\frac{1}{m}$ factor denotes the probability which with each state appears in the ensemble.

\begin{lem}
\label{lem2}
If Alice starts the measurement protocol to distinguish $m$ MES by LOCC, the post measurement reduced states (PMRS) on her side are completely indistinguishable.
\end{lem}

\begin{proof}
Since $\{\ket{\psi_i}\}_{i=1}^{m}$ are MES, the corresponding reduced states on Alice's subsystem are maximally mixed, i.e., $\rho^{(A)}_i = \frac{1}{d}\mathbb{1}_A$. As a result of the measurement, the states on Alice's subsystem transform as $\rho^{(A)}_i$ $=$ $\dfrac{1}{d}$ $\mathbb{1}_A$ $\longrightarrow$ $ \rho^{(A)}_{\s{i,\alpha}}$ $\propto$ $ K_\alpha K_\alpha^\dag$, $\forall \; 1 \leq i \leq m$. This implies that (even) after the first measurement, the PMRS on Alice's side are completely indistinguishable.
\end{proof}

For the post-measurement joint states $\{ \ket{\psi_{\s{i,\alpha}}} \}_{i=1}^{m}$ to still be distinguishable, the indistinguishability of PMRS on Alice's side imposes constraints on the average PMRS on Bob's side. This is made clear in theorem \ref{thm1}.

\begin{thm}
\label{thm1}
If the PMRS on Alice's side are completely indistinguishable, the von Neuman entropy of the average PMRS on Bob's side has to be at least $log_2 m \; bit$ for the states to be perfectly distinguishable by LOCC.
\end{thm}

\begin{proof}

The Holevo-like upper bound for the locally accessible information of the set of states $\{ \ket{\psi_{\s{i,\alpha}}} \}_{i=1}^{m}$ is given by  \cite{B03}
\begin{align}
\label{LOCCacc2}
I^{LOCC}_{acc}  \leq & \;  S \left( \rho^{(A)}_{\s{\alpha}} \right) + S \left( \rho^{(B)}_{\s{\alpha}} \right) \notag  \\ & -   Max \; \left\{ \frac{1}{m}\sum_{i=1}^{m} S \left( \rho^{(X)}_{\s{i,\alpha}} \right) \; : \; X=A,B \right\}.
\end{align}
Since, the post-measurement states $\ket{\psi_{i,\alpha}}$ are all pure, we have the following: $$S \left(\rho^{(A)}_{\s{i,\alpha}} \right)= S \left(\rho^{(B)}_{\s{i,\alpha}} \right), \; \forall \; 1 \leq i \leq m.$$ Also, that Alice's PMRS are completely indistinguishable (Lemma \ref{lem2}) implies that $S\left(\rho^{(A)}_{\s{\alpha}}\right) = S\left(\rho^{(A)}_{\s{i,\alpha}}\right)$, $\forall \; 1 \leq i \leq m$. This implies that $I^{LOCC}_{acc} \leq S \left( \rho^{(B)}_{\s{\alpha}} \right)$. Since we need to distinguish between $m$ different states, we require that $I^{LOCC}_{acc}$ be at least $log_2 m$ bits. Then the aforementioned inequality tells us that $S(\rho^{(B)}_{\s{\alpha}})$ has to be at least $log_2m$ bits.
\end{proof}

With respect to the standard ONB $\{ \ket{j}_A \}_{j=1}^{d}$ of Alice's system, every MES $\ket{\psi_i}$ from the shared ensemble $\{ \ket{\psi_i} \}_{i=1}^{m}$ can be expressed as

\begin{equation}
\label{schmidt1}
\ket{\psi_i} = \dfrac{1}{\sqrt{d}} \sum_{j=1}^{d}  \ket{j}_{A}  \ket{b_j^{(i)}}_B,
\end{equation}

where $\{ \ket{b_j^{(i)}}_B\}_{j=1}^{d}$ is an ONB for Bob's system for each $i = 1, 2, \cdots, m$. \medbreak

The $i$-th PMRS on Bob's side is then given by
\begin{align}
\label{BT}
 \rho^{(B)}_{\s{i,\alpha}} & =  \frac{1}{Tr\left(K_\alpha^\dag K_\alpha\right)}\sum_{j,k=1}^{d} \pk{j}{K_\alpha^\dag K_\alpha}{k} \ketbra{b_k^{(i)}}{b_j^{(i)}} \notag \\ &  = U_i^{(B)} \dfrac{  K_\alpha^T K_\alpha^* } {Tr\left( K_\alpha^T K_\alpha^*\right)} {U_i^{(B)}}^\dag,
\end{align}

where $U_i^{(B)}$ are local unitaries on Bob's side, such that $U_i^{(B)} \ket{j}_B = \ket{b_j^{(i)}}_B$, for $j =1,2, \cdots, d$, where $i = 1,2, \cdots, m$, and where $K_\alpha^T K_\alpha^*$ are operators on Bob's system, whose matrix elements with respect to the ONB $\{ \ket{j}_{B} \}_{j=1}^{d}$ are the same as the complex conjugate of matrix elements of Alice's POVM effect $K_\alpha^\dag K_\alpha$ with respect to the ONB $\{ \ket{j}_A \}_{j=1}^{d}$. The average PMRS corresponding to the set on Bob's side is thus given by

\begin{equation}
\label{BR}
\rho^{(B)}_{\s{\alpha}} = \sum_{i=1}^{m} \dfrac{1}{m} U_i^{(B)} \dfrac{ K_\alpha^T K_\alpha^* } {Tr\left( K_\alpha^T K_\alpha^*\right)} {U_i^{(B)}}^\dag.
\end{equation} We require that $\rho^{(B)}_{\s{\alpha}}$ satisfies theorem \eqref{thm1}. This requirement puts a constraint on Alice's starting measurement.

We already know one constraint on Alice's starting measurement, i.e.,  it should be OP. Hence whenever $i \neq j$,

\begin{equation}
\label{orthocondition}
\pk{\psi_i}{K_\alpha^\dag K_\alpha \otimes \mathbb{1}}{\psi_j}= 0.
\end{equation}

It is easy to see that condition \eqref{orthocondition} should be subsumed in the requirement that $\rho^{(B)}_{\s{\alpha}}$ should satisfy theorem \eqref{thm1}.

Consider the special case when $m=d$.
\begin{cor}
If $m = d$ in Theorem \eqref{thm1}, then the average PMRS on Bob's side has to be maximally mixed.
\end{cor}

\begin{proof}
When $m=d$, we require $log_2d \; bit$ of information to distinguish between $d$ states. The maximal value that $S ( \rho^{(B)}_{\s{\alpha}} )$ can take is $log_2d$ and it can take this value \emph{only} when $\rho^{(B)}_{\s{\alpha}}$ is maximally mixed.
\end{proof}
Thus, \emph{requiring that $S(\rho^{(B)}_{\s{\alpha}})$ be at least $log_2(d) \; bit$ implies that $\rho^{(B)}_{\s{\alpha}}$ has to be a maximally mixed state}, i.e., we require that

\begin{equation}
\label{neccond1}
\sum_{i=1}^{d} \dfrac{1}{d} U_i^{(B)} \dfrac{K_\alpha^T  K_\alpha^*} {Tr\left( K_\alpha^T K_\alpha^*\right)} {U_i^{(B)}}^\dag = \dfrac{1}{d} \mathbb{1}_d.
\end{equation}

The matrix equation \eqref{neccond1} is linear in the matrix elements of $\frac{1}{Tr\left( K_\alpha^\dag K_\alpha \right)}K_\alpha^\dag K_\alpha$, which are also the unknowns of the equation. Hence, we solve equation \eqref{neccond1} to obtain the solution space for the matrix elements of $\frac{1}{Tr\left( K_\alpha^\dag K_\alpha \right)}K_\alpha^\dag K_\alpha$. 

Note that before any measurement, $I^{LOCC}_{acc}$ for the states $\{ \ket{\psi_i} \}_{i=1}^{m}$ is $log_2 d$ bits. The solution of condition \eqref{neccond1} provides us those $K_\alpha^\dag K_\alpha$, for which the $I^{LOCC}_{acc}$ of the post-measurement states $\{ \ket{\psi_{i,\alpha} }\}_{i=1}^m $ is also $log_2 d$ bits. Hence, if the only solution of equation \eqref{neccond1} is that $$ K_\alpha^\dag K_\alpha \propto
 \mathbb{1}_A,$$ then that implies that any non-trivial measurement by Alice will decrease $I^{LOCC}_{acc}$ of the post-measurement states lower than $log_2 d$ bits, which implies that the post-measurements states aren't locally distinguishable anymore. Thus the necessary condition for local indistinguishability is given as follows.

\textbf{R (Necessary Condition): }\emph{If the solution of the condition \eqref{neccond1} is that $K_\alpha^\dag K_\alpha  \propto \mathbb{1}_A$, the states $\{ \ket{\psi_i} \}_{i=1}^{m}$ are indistinguishable by LOCC. If not, then the states may still be distinguishable by LOCC.} 

\section{Perfect Local Distinguishability of Four Generalized Bell States in $\mathbf{\C{4}}$}
\label{secensgen}

The necessary condition, \textbf{R} \eqref{neccond1} has to be tested for protocols initiated by both Alice and Bob, separately. Consider the case of generalized Bell states.

\begin{defn}
\label{def1}
Generalized Bell states are MES in $\C{d}$ of the form

\begin{equation}
\label{GenBell}
\ket{\psi_{nm}^{(d)}} \equiv \dfrac{1}{\sqrt{d}}\sum_{j=0}^{d-1}e^{\frac{2 \pi i j n}{d}} \ket{j}_A\ket{j \oplus_d m}_B,
\end{equation}

where $n,m \in \{ 0, 1, \cdots, d-1 \}$ and where $\{ \ket{j}_A \}_{j=0}^{d-1}$ is an ONB for Alice's subsystem and $\{ \ket{j}_B \}_{j=0}^{d-1}$ is an ONB for Bob's subsystem. Here $j \oplus_d m \equiv (j + m) mod\; d$.
\end{defn}

Note that $\braket{\psi_{n'm'}}{\psi_{nm}}= \delta_{n'n}\delta_{m'm}$, $\forall \; n,n',m,m' \in \{ 0, 1, \cdots, d-1 \}$.

When $d=4$, $\binom{16}{4}$ sets of four generalized Bell states can be partitioned into $122$ distinct equivalence classes, where all sets in the same equivalence class differ from each other by action of a local unitary (see Appendix \ref{122} for details). Thus, if states in any one set in an equivalence class are locally (in)distinguishable, then so are the states of any other set in the same equivalence class. We tested the condition \textbf{R} \eqref{neccond1} on these $122$ equivalence classes, and found that $39$ of them fail the necessary condition, and are hence locally indistinguishable. An explicit proof of one such a set is given in the following example \ref{ens2}.

\begin{example}
\label{ens2}
The states $\ket{\psi_{00}^{(4)}}$, $\ket{\psi_{11}^{(4)}}$, $\ket{\psi_{31}^{(4)}}$, $\ket{\psi_{32}^{(4)}}$ are locally indistinguishable.
\end{example}

\begin{proof}
In reference \cite{S11} (example 1, p 6), it has already been shown that the given set of states are indistinguishable by one-way LOCC when only projective measurements are used in the LOCC protocol. Here we will generalize the result for all possible LOCC protocols. Also, we show that the condition \eqref{neccond1} is stronger than the OP condition \eqref{orthocondition}.

Let Alice commense the protocol by applying a measurement, whose Kraus operators are $\{K_i\}_{i=1}^{n}$ on her subsystem, and obtain the $\alpha$-th outcome. We impose the conditions \eqref{orthocondition} on $K_\alpha^\dag K_\alpha$. 

The orthogonality preserving condition \eqref{orthocondition} is given by:

\begin{subequations}
\label{nmn'm'}
\begin{align}
\label{0011}
& \pk{\psi_{00}^{(4)}}{\left( K_\alpha^\dag K_\alpha \otimes \mathbb{1} \right)}{\psi_{11}^{(4)}}=0,\\
\label{0031}
&\pk{\psi_{00}^{(4)}}{\left( K_\alpha^\dag K_\alpha \otimes \mathbb{1} \right)}{\psi_{31}^{(4)}}=0,\\
\label{0032}
&\pk{\psi_{00}^{(4)}}{\left( K_\alpha^\dag K_\alpha \otimes \mathbb{1} \right)}{\psi_{32}^{(4)}}=0,\\
\label{1131}
&\pk{\psi_{11}^{(4)}}{\left( K_\alpha^\dag K_\alpha \otimes \mathbb{1} \right)}{\psi_{31}^{(4)}}=0,\\
\label{1132}
&\pk{\psi_{11}^{(4)}}{\left( K_\alpha^\dag K_\alpha \otimes \mathbb{1} \right)}{\psi_{32}^{(4)}}=0,\\
\label{3132}
&\pk{\psi_{31}^{(4)}}{\left( K_\alpha^\dag K_\alpha \otimes \mathbb{1} \right)}{\psi_{32}^{(4)}}=0.
\end{align}
\end{subequations}

Let the spectral decomposition of $K_\alpha^\dag K_\alpha$ be given by 

\begin{equation}
\label{KKSD}
K_\alpha^\dag K_\alpha = \proj{u} + \proj{v}+\proj{w} + \proj{x},
\end{equation}
where $\braket{u}{v}=\braket{u}{w}=\braket{u}{x}=\braket{v}{w}=\braket{v}{x}=\braket{w}{x}=0$, but $\ket{u}$, $\ket{v}$, $\ket{w}$ and $\ket{x}$ aren't normalized. 

Additionally, let's $\ket{u}$, $\ket{v}$, $\ket{w}$ and $\ket{x}$ have the following expansions in the standard ONB.

\begin{subequations}
\label{uvwx}
\begin{align}
\label{u}
&\ket{u} = \sum_{i=0}^{3}u_i \ket{i}, \\
\label{v}
&\ket{v} = \sum_{i=0}^{3}v_i \ket{i}, \\
\label{w}
&\ket{w} = \sum_{i=0}^{3}w_i \ket{i}, \\
\label{x}
&\ket{x} = \sum_{i=0}^{3}x_i \ket{i}.
\end{align}
\end{subequations}

Then equations \eqref{0011}-\eqref{3132} respectively become (using the definition \eqref{def1})

\begin{subequations}
\label{nmn'm'1}
\small
\begin{align}
\label{00111}
  & (u_1u_0^*+v_1v_0^* + w_1w_0^*+ x_1x_0^*) \notag \\ + & i (u_2u_1^*+v_2v_1^* + w_2w_1^*+ x_2x_1^*) \notag \\
- & (u_3u_2^*+v_3v_2^* + w_3w_2^*+ x_3x_2^*) \notag \\ - & i (u_0u_3^*+v_0v_3^* + w_0w_3^*+ x_0x_3^*) = 0, 
\end{align}

\begin{align}
\label{00311}
  & (u_1u_0^*+v_1v_0^* + w_1w_0^*+ x_1x_0^*) \notag \\ - & i (u_2u_1^*+v_2v_1^* + w_2w_1^*+ x_2x_1^*)  \notag \\
- & (u_3u_2^*+v_3v_2^* + w_3w_2^*+ x_3x_2^*) \notag \\ + & i(u_0u_3^*+v_0v_3^* + w_0w_3^*+ x_0x_3^*) = 0, 
\end{align}

\begin{align}
\label{00321}
  & (u_2u_0^*+v_2v_0^* + w_2w_0^*+ x_2x_0^*)  \notag \\- & i (u_3u_1^*+v_3v_1^* + w_3w_1^*+ x_3x_1^*) \notag \\
- & (u_0u_2^*+v_0v_2^* + w_0w_2^*+ x_0x_2^*)  \notag \\+  & i (u_1u_3^*+v_1v_3^* + w_1w_3^*+ x_1x_3^*) = 0, 
\end{align}

\begin{align}
\label{11311}
  & (u_0u_0^*+v_0v_0^* + w_0w_0^*+ x_0x_0^*) \notag \\ -  &(u_1u_1^*+v_1v_1^* + w_1w_1^*+ x_1x_1^*) \notag \\ 
+ &  (u_2u_2^*+v_2v_2^* + w_2w_2^*+ x_2x_2^*) \notag \\ -  &(u_3u_3^*+v_3 v_3^* + w_3w_3^*+ x_3x_3^*)  = 0, 
\end{align}

\begin{align}
\label{11321}
  &(u_1u_0^*+v_1v_0^* + w_1w_0^*+ x_1x_0^*) \notag \\ - &  (u_2u_1^*+v_2v_1^* + w_2w_1^*+ x_2x_1^*) \notag \\ 
+ &  (u_3u_2^*+v_3v_2^* + w_3w_2^*+ x_3x_2^*)  \notag \\- & (u_0u_3^*+v_0 v_3^* + w_0w_3^*+ x_0x_3^*)  = 0, 
\end{align}

\begin{align}
\label{31321}
  &(u_1u_0^*+v_1v_0^* + w_1w_0^*+ x_1x_0^*)  \notag \\+ & (u_2u_1^*+v_2v_1^* + w_2w_1^*+ x_2x_1^*) \notag \\
+ &  (u_3u_2^*+v_3v_2^* + w_3w_2^*+ x_3x_2^*) \notag \\ +& (u_0u_3^*+v_0 v_3^* + w_0w_3^*+ x_0x_3^*)  = 0. 
\end{align}

\end{subequations}

Expand $K_\alpha^\dag K_\alpha$ in the $\ketbra{i}{j}$ basis; then equations \eqref{u}-\eqref{x} allow us to define
\small
\begin{equation}
\label{xi}
 \xi_{ij} \equiv \left( K_\alpha^\dag K_\alpha\right)_{ij} = u_iu_j^*+ v_iv_j^*+ w_iw_j^* + x_ix_j^* , \; \forall \; 0 \leq i,j \leq 3. 
\end{equation}
\normalsize
Using equation \eqref{xi}, equations \eqref{00111}-\eqref{31321} take a condensed form.

\begin{subequations}

\begin{equation}
\label{00112}
\begin{matrix} \begin{pmatrix}
\xi_{01}^* & \xi_{12}^* & \xi_{23}^* & \xi_{30}^* \end{pmatrix} \\ \mbox{} \\ \mbox{} \\ \mbox{} \end{matrix}
\begin{pmatrix}
1 \\ 
i \\
-1\\
-i
\end{pmatrix}=0.
\end{equation}

\begin{equation}
\label{00312}
\begin{matrix} \begin{pmatrix}
\xi_{01}^* & \xi_{12}^* & \xi_{23}^* & \xi_{30}^* \end{pmatrix} \\ \mbox{} \\ \mbox{} \\ \mbox{} \end{matrix}
\begin{pmatrix}
\;1\\ 
-i\\
-1\\
\;i
\end{pmatrix}=0.
\end{equation}

\begin{equation}
\label{00322}
\begin{matrix} \begin{pmatrix}
\xi_{02}^* & \xi_{13}^* & \xi_{02} & \xi_{13} \end{pmatrix} \\ \mbox{} \\ \mbox{} \\ \mbox{} \end{matrix}
\begin{pmatrix}
\;1\\ 
-i\\
-1\\
\;i
\end{pmatrix}=0.
\end{equation}

\begin{equation}
\label{11312}
\begin{matrix} \begin{pmatrix}
\xi_{00} & \xi_{11} & \xi_{22} & \xi_{33} \end{pmatrix} \\ \mbox{} \\ \mbox{} \\ \mbox{} \end{matrix}
\begin{pmatrix}
\;1\\ 
-1\\
\;1\\
-1
\end{pmatrix}=0.
\end{equation}

\begin{equation}
\label{11322}
\begin{matrix} \begin{pmatrix}
\xi_{01}^* & \xi_{12}^* & \xi_{23}^* & \xi_{30}^*  \end{pmatrix} \\ \mbox{} \\ \mbox{} \\ \mbox{} \end{matrix}
\begin{pmatrix}
\;1\\ 
-1\\
\;1\\
-1
\end{pmatrix}=0.
\end{equation}

\begin{equation}
\label{31322}
\begin{matrix} \begin{pmatrix}
\xi_{01}^* & \xi_{12}^* & \xi_{23}^* & \xi_{30}^*  \end{pmatrix} \\ \mbox{} \\ \mbox{} \\ \mbox{} \end{matrix}
\begin{pmatrix}
1\\ 
1\\
1\\
1
\end{pmatrix}=0.
\end{equation}

\end{subequations}

Equations \eqref{00112} , \eqref{00312}, \eqref{11322} and \eqref{31322} collectively imply that 

\begin{equation}
\label{offdiag1}
(\xi_{01},\xi_{12}, \xi_{23},\xi_{30})= 0. 
\end{equation}

Equation \eqref{11312} implies that 
\begin{align}
\label{diag}
& (\xi_{00},\xi_{11}, \xi_{22},\xi_{33}) \notag \\  = & \; a_0 (1,1,1,1) + a_1(1,1,-1,-1) + a_2(1,-1,-1,1),
\end{align}
where $a_0$, $a_1$ and $a_2$ are real. This is because $\xi_{ii}$ are diagonal matrix elements of $K_\alpha^\dag K_\alpha$.

Equation \eqref{00312} implies that $(\xi_{02},\xi_{13}, \xi_{20},\xi_{31})$ has to be of the form 

\begin{align}
\label{offdiag2}
& (\xi_{02},\xi_{13}, \xi_{20},\xi_{31}) \notag \\  = & \; b_0 (1,1,1,1) + b_1(1,i,-1,-i) + b_2(1,-1,1,-1).
\end{align}

Since $K_\alpha^\dag K_\alpha$ is hermitian, $\xi_{ij} = \xi_{ji}^*$. This implies that $(b_0 - b_1 + b_2)^* = b_0 + b_1 + b_2$ (using $i=0,j=2$) and $(b_0+ib_1 - b_2)^* = b_0 -i b_1 -b_2$ (using $i=1,j=3$) and these imply that $b_1 = 0$ and $b_0$ and $b_2$ are real.  

Thus putting the constraints imposed by equations \eqref{offdiag1}, \eqref{diag} and \eqref{offdiag2}, tells us that in the $\ketbra{i}{j}$ basis $K_\alpha^\dag K_\alpha$ is given by equation:
\small
\begin{align}
\label{KKpc}
& K_\alpha^\dag K_\alpha   = \notag \\ &  \begin{pmatrix}
                          a_0 + a_1 + a_2 & 0                & b_0 + b_2       & 0               \\
                          0               &  a_0 + a_1 - a_2 & 0               & b_0 - b_2       \\
                          b_0 + b_2       & 0                & a_0 - a_1 - a_2 & 0               \\
                          0               & b_0 - b_2        & 0               & a_0 - a_1 + a_2
                         \end{pmatrix}.
\end{align} 
\normalsize
The eigensystem of $K_\alpha^\dag K_\alpha$ in equation \eqref{KKpc} is given in the following table 
\begin{table}[H]
\label{tabeig}
 \begin{tabular}{ | c | c | }
 \hline
  Eigenvalue                 & Eigenvector  \\  \hline 
  $\lambda_u =  a_0 + \mu_0$ & $\ket{u}  =  N_u \begin{pmatrix}
			      a_1 + a_2 + \mu_0 \\ 
			      0 \\	
			      b_0 + b_1 \\
			      0                 
			      \end{pmatrix}  \equiv \begin{pmatrix}
			      cos \frac{\zeta}{2} \\ 
			      0 \\
			      sin \frac{\zeta}{2} \\
			      0                 
			      \end{pmatrix}$ \\
			      \hline 
  $\lambda_v =  a_0 - \mu_0$ & $\ket{v}  = N_v \begin{pmatrix}
                  a_1 + a_2 - \mu_0 \\ 
                  0 \\
                  b_0 + b_1 \\
                  0                 
                 \end{pmatrix}  \equiv \begin{pmatrix}
                  -sin \frac{\zeta}{2} \\ 
                  0 \\
                  cos \frac{\zeta}{2} \\
                  0                 
                 \end{pmatrix}$ \\
			      \hline 
  $\lambda_w =  a_0 + \mu_1$ & $\ket{w}  = N_w \begin{pmatrix}
                  (a_1 - a_2) - \mu_1 \\ 
                  0 \\
                  b_0 - b_1 \\
                  0                 
                 \end{pmatrix}  \equiv \begin{pmatrix}
                  cos \frac{\eta}{2} \\ 
                  0 \\
                  sin \frac{\eta}{2} \\
                  0                 
                 \end{pmatrix}$ \\
			      \hline 
  $\lambda_x =  a_0 + \mu_1$ & $ \ket{x}  = N_x \begin{pmatrix}
                  (a_1 - a_2) - \mu_1 \\ 
                  0 \\
                  b_0 - b_1 \\
                  0                 
                 \end{pmatrix}  \equiv \begin{pmatrix}
                  -sin \frac{\eta}{2} \\ 
                  0 \\
                  cos \frac{\eta}{2} \\
                  0                 
                 \end{pmatrix}$ \\ 
			      \hline 
 \end{tabular}
 \caption{Eigenvalues and Eigenvectors of $K_\alpha^\dag K_\alpha$.}
\end{table}
where $\mu_0$ and $\mu_1$ are given by $ \sqrt{(a_1+a_2)^2 + (b_0 + b_2)^2}$ and $\sqrt{(a_1-a_2)^2 + (b_0 - b_2)^2}$ and $N_u$, $N_v$, $N_w$ and $N_x$ are normalization factors. For $K_\alpha^\dag K_\alpha$ to be a positive semidefinite operator it is necessary that $a_0 \geq |\mu_0|,|\mu_1|$. 

Using the table above,
\begin{align}
\label{KKexp}
K_\alpha^\dag K_\alpha  = &   \; a_0 \mathbb{1} \notag \\
                          &  +    \mu_0 \begin{pmatrix}
                                                                              cos \zeta & 0 & sin \zeta  & 0 \\
                                                                              0         & 0 & 0          & 0 \\
                                                                              sin \zeta & 0 & -cos \zeta & 0 \\
                                                                              0         & 0 & 0          & 0 
\end{pmatrix}  \notag \\
               & +  \mu_1 \begin{pmatrix}
                                                                             0 &  0        & 0 & 0         \\
                                                                             0 &  cos \eta & 0 & sin \eta  \\
                                                                             0 &  0        & 0 & 0         \\
                                                                             0 &  sin \eta & 0 & -cos \eta
\end{pmatrix}.
\end{align}
Imposing the OP condition \eqref{orthocondition} doesn't give us any conclusion about the local (in)distinguishability of the states  $\ket{\psi_{00}^{(4)}}$, $\ket{\psi_{11}^{(4)}}$, $\ket{\psi_{31}^{(4)}}$, $\ket{\psi_{32}^{(4)}}$, since the solution for $K_\alpha^\dag K_\alpha$ (equation \eqref{KKexp}) is not a multiple of the identity, i.e., the measurement isn't constrained to be trivial. Hence, at this point we do not know if the states are distinguishable or not. 

We now obtain the post-measurement joint states $\ket{\psi^{(4)}_{\s{\alpha,nm}}}$ using necessary condition \eqref{orthocondition} for OP. 

$K_\alpha^\dag K_\alpha$ enables us to determine $K_\alpha$ upto a left-unitary, i.e., $K_\alpha = U \sqrt{K_\alpha^\dag K_\alpha}$, where $U$ is a $4 \times 4$ unitary matrix. This unitary $U$ is irrelevant because physically it implies Alice performing a unitary after her measurement and we know that such a unitary transformation on Alice's side (or Bob's side) doesn't alter the local distinguishability of the set of states. Hence we can assume the $U = \mathbb{1}_4$. Using the above table of eigenvalues and eigenvectors, we get that $K_\alpha \left(\equiv \sqrt{K_\alpha^\dag K_\alpha}\right)$ is given by \small
\begin{widetext}

\begin{equation}
\label{K}
\begin{split}
& K_\alpha =   \\ 
& \sqrt{a_0 +  \mu_0} \left( cos \frac{\zeta}{2}\ket{0} + sin \frac{\zeta}{2} \ket{2}  \right) \left( cos \frac{\zeta}{2}\bra{0} + sin \frac{\zeta}{2} \bra{2}  \right)  \\
 & +              \sqrt{a_0 -  \mu_0} \left( -sin \frac{\zeta}{2}\ket{0} + cos \frac{\zeta}{2} \ket{2}  \right) \left( -sin \frac{\zeta}{2}\bra{0} + cos \frac{\zeta}{2} \bra{2}  \right)  \\ 
 & +              \sqrt{a_0 +  \mu_1} \left( cos \frac{\eta}{2}\ket{1} + sin \frac{\eta}{2} \ket{3}  \right) \left( cos \frac{\eta}{2}\bra{1} + sin \frac{\eta}{2} \bra{3}  \right)  \\
 & +              \sqrt{a_0 - \mu_1} \left( -sin \frac{\eta}{2}\ket{1} + cos \frac{\eta}{2} \ket{3}  \right) \left( -sin \frac{\eta}{2}\bra{1} + cos \frac{\eta}{2} \bra{3}  \right).
\end{split}
\end{equation}

Using equation \eqref{K} we now give the Schmidt decomposition of the states $\ket{\psi^{(4)}_{\s{\alpha,00}}}, \ket{\psi^{(4)}_{\s{\alpha,11}}}, \ket{\psi^{(4)}_{\s{\alpha,31}}}, \ket{\psi^{(4)}_{\s{\alpha,32}}} $.
\begin{subequations}
\label{phis}
\begin{equation}
\label{phi00}
\begin{split}
\ket{\psi^{(4)}_{\s{\alpha,00}}}= ~& \dfrac{1}{2}\sqrt{1+\frac{\mu_0}{a_0} } \; \ket{\chi}_A        \left( cos \frac{\zeta}{2}\ket{0}_B + sin \frac{\zeta}{2}\ket{2}_B \right)  
 +  \dfrac{1}{2}\sqrt{1-\frac{\mu_0}{a_0} } \; \ket{\kappa}_A           \left(-sin \frac{\zeta}{2}\ket{0}_B + cos \frac{\zeta}{2}\ket{2}_B \right)   \\
+ &  \dfrac{1}{2}\sqrt{1+\frac{\mu_1}{a_0} } \; \ket{\omega}_A \left( cos \frac{\eta}{2}\ket{1}_B + sin \frac{\eta}{2}\ket{3}_B \right)   
 +   \dfrac{1}{2}\sqrt{1-\frac{\mu_1}{a_0} } \; \ket{\tau}_A           \left(-sin \frac{\eta}{2}\ket{1}_B + cos \frac{\eta}{2}\ket{3}_B \right),  
\end{split}
\end{equation}
\begin{equation}
\label{phi11}
\begin{split}
\ket{\psi^{(4)}_{\s{\alpha,11}}}=  ~& \dfrac{1}{2}\sqrt{1+\frac{\mu_0}{a_0} } \; \ket{\chi}_A        \left( cos \frac{\zeta}{2}\ket{1}_B - sin \frac{\zeta}{2}\ket{3}_B \right)   
 -   \dfrac{1}{2}\sqrt{1-\frac{\mu_0}{a_0} } \; \ket{\kappa}_A           \left(sin \frac{\zeta}{2}\ket{1}_B + cos \frac{\zeta}{2}\ket{3}_B \right)  \\  
 + & i  \dfrac{1}{2}\sqrt{1+\frac{\mu_1}{a_0} } \; \ket{\omega}_A  \left( - sin \frac{\eta}{2}\ket{0}_B + cos \frac{\eta}{2}\ket{2}_B \right)   
 - i  \dfrac{1}{2}\sqrt{1-\frac{\mu_1}{a_0} } \; \ket{\tau}_A           \left( cos \frac{\eta}{2}\ket{0}_B+  sin \frac{\eta}{2}\ket{2}_B \right),  
\end{split}
\end{equation}
\begin{equation}
\label{phi31}
\begin{split}
 \ket{\psi^{(4)}_{\s{\alpha,31}}}= ~ & \dfrac{1}{2}\sqrt{1+\frac{\mu_0}{a_0} } \; \ket{\chi}_A        \left( cos \frac{\zeta}{2}\ket{1}_B - sin \frac{\zeta}{2}\ket{3}_B \right)    
 -   \dfrac{1}{2}\sqrt{1-\frac{\mu_0}{a_0} } \; \ket{\kappa}_A           \left(sin \frac{\zeta}{2}\ket{1}_B +  cos \frac{\zeta}{2}\ket{3}_B \right)   \\  
- & i  \dfrac{1}{2}\sqrt{1+\frac{\mu_1}{a_0} } \; \ket{\omega}_A  \left( - sin \frac{\eta}{2}\ket{0}_B + cos \frac{\eta}{2}\ket{2}_B \right) 
 + i  \dfrac{1}{2}\sqrt{1-\frac{\mu_1}{a_0} } \; \ket{\tau}_A           \left( cos \frac{\eta}{2}\ket{0}_B+  sin \frac{\eta}{2}\ket{2}_B \right),  
\end{split}
\end{equation}
\begin{equation}
\begin{split}
\label{phi32}
\ket{\psi^{(4)}_{\s{\alpha,32}}}= ~ & \dfrac{1}{2}\sqrt{1+\frac{\mu_0}{a_0} } \; \ket{\chi}_A        \left(- sin \frac{\zeta}{2}\ket{0}_B + cos \frac{\zeta}{2}\ket{2}_B  \right)  
-   \dfrac{1}{2}\sqrt{1-\frac{\mu_0}{a_0} } \; \ket{\kappa}_A           \left(   cos \frac{\zeta}{2}\ket{0}_B+ sin \frac{\zeta}{2}\ket{2}_B \right)   \\  
 + & i  \dfrac{1}{2}\sqrt{1+\frac{\mu_1}{a_0} } \; \ket{\omega}_A  \left( sin \frac{\eta}{2}\ket{1}_B - cos \frac{\eta}{2}\ket{3}_B \right)  
 +  i  \dfrac{1}{2}\sqrt{1-\frac{\mu_1}{a_0} } \; \ket{\tau}_A           \left( cos \frac{\eta}{2}\ket{1}_B+  sin \frac{\eta}{2}\ket{3}_B \right),  
\end{split}
\end{equation}
\end{subequations}
\end{widetext}
where \smallbreak
$\ket{\chi}_A = \left( cos \frac{\zeta}{2}\ket{0}_A + sin \frac{\zeta}{2}\ket{2}_A \right)$, \smallbreak
$\ket{\kappa}_A = \left(-sin \frac{\zeta}{2}\ket{0}_A + cos \frac{\zeta}{2}\ket{2}_A \right) $, \smallbreak
$\ket{\omega}_A =  \left( cos \frac{\eta}{2}\ket{1}_A + sin \frac{\eta}{2}\ket{3}_A \right)$, \smallbreak
$\ket{\tau}_A =  \left(-sin \frac{\eta}{2}\ket{1}_A + cos \frac{\eta}{2}\ket{3}_A \right) $ \medbreak
are vectors of an ONB on Alice's subsystem. It is easy to check that the states  $\ket{\psi^{(4)}_{\s{\alpha,00}}}, \ket{\psi^{(4)}_{\s{\alpha,11}}}, \ket{\psi^{(4)}_{\s{\alpha,31}}}, \ket{\psi^{(4)}_{\s{\alpha,32}}} $ are pairwise orthogonal.

From equation \eqref{phis} it can be easily seen that the spectra of $\rho^{(A)}_{\s{\alpha}}$, $\rho^{(A)}_{\s{\alpha,00}}$, $\rho^{(A)}_{\s{\alpha,11}}$, $\rho^{(A)}_{\s{\alpha,31}}$, $\rho^{(A)}_{\s{\alpha,32}}$, $\rho^{(B)}_{\s{\alpha,00}}$, $\rho^{(B)}_{\s{\alpha,11}}$, $\rho^{(B)}_{\s{\alpha,31}}$, $\rho^{(B)}_{\s{\alpha,32}}$ are the same, and the common spectra is of the form \small$\left\{ \frac{1+\frac{\mu_0}{a_0}}{4} \right. $, $\frac{1-\frac{\mu_0}{a_0}}{4}$, $\frac{1+\frac{\mu_1}{a_0}}{4}$, $\left. \frac{1-\frac{\mu_1}{a_0}}{4} \right\}$\normalsize. The equality of the spectra of the aforementioned states was earlier predicted by theorem \eqref{thm1}. The spectrum of \small $\rho^{(B)}_{\s{\alpha}}$ is given by $\left\{ \frac{1+\frac{\mu_0}{2a_0}}{4}, \frac{1-\frac{\mu_0}{2a_0}}{4}, \frac{1+\frac{\mu_1}{2a_0}}{4}, \frac{1-\frac{\mu_1}{2a_0}}{4} \right\} $\normalsize. Substituting these quantities in the LHS of the inequality \eqref{LOCCacc2} we get that
\begin{align}
\label{LOCCacc1}
I^{LOCC}_{acc} &  \leq S\left(\rho^{(B)}_{\s{\alpha}}\right) \notag \\ &   =  H\left(\dfrac{1+\frac{\mu_0}{2a_0}}{4},\dfrac{1-\frac{\mu_0}{2a_0}}{4},\dfrac{1+\frac{\mu_1}{2a_0}}{4},\dfrac{1-\frac{\mu_1}{2a_0}}{4}\right),
\end{align} 
where $H\left(\frac{1+\frac{\mu_0}{2a_0}}{4},\frac{1-\frac{\mu_0}{2a_0}}{4},\frac{1+\frac{\mu_1}{2a_0}}{4},\frac{1-\frac{\mu_1}{2a_0}}{4}\right)$ is the Shannon entropy for $\left\{ \frac{1+\frac{\mu_0}{2a_0}}{4},\frac{1-\frac{\mu_0}{2a_0}}{4},\frac{1+\frac{\mu_1}{2a_0}}{4},\frac{1-\frac{\mu_1}{2a_0}}{4} \right\}$.

Thus we see that unless $\mu_0$ $=$ $\mu_1$ $=$ $0$, $H(  \frac{1}{4}(1+\frac{\mu_0}{2a_0})$, $\frac{1}{4}(1-\frac{\mu_0}{2a_0})$,  $\frac{1}{4}(1+\frac{\mu_1}{2a_0})$, $\frac{1}{4}(1-\frac{\mu_1}{2a_0}))$ $< 2 \; bit$. Thus if $\mu_0 \neq 0$ or $\mu_1 \neq 0$, the locally accessible information of the set $\{ \ket{\psi_{\s{\alpha,00}}},\ket{\psi_{\s{\alpha,11}}},\ket{\psi_{\s{\alpha,31}}},\ket{\psi_{\s{\alpha,32}}}\}$ is lower than $2 \;  bit$, meaning that there is no LOCC protocol that Alice and Bob can use to perfectly distinguish between states in the set. On the other hand if $\mu_0 = \mu_1 =0$, Alice's POVM is a trivial one (see equation \eqref{KKexp}). Thus the states  $\ket{\psi_{00}^{(4)}}$, $\ket{\psi_{11}^{(4)}}$, $\ket{\psi_{31}^{(4)}}$, $\ket{\psi_{32}^{(4)}}$ fail to satisfy the necessary condition, \textbf{R}. It is significant to note that the local indistinguishability of these states was established \emph{only} after it was demanded that $S(\rho^{(B)}_{\s{\alpha}}) = 2 \; bit$, in accordance with 
theorem \eqref{thm1}. Thus, this also shows that condition \eqref{neccond1}, which is equivalent to theorem \eqref{thm1} when $m=d$, is stronger than the OP condition \eqref{orthocondition}.

This shows us that there is no LOCC protocol to perfectly distinguish the states in the set, if Alice starts the protocol. Similarly it can be shown that there is no LOCC protocol to perfectly distinguish the states of the set, if Bob starts the protocol; the arguments to establish this follow the same sequence of reasoning as the arguments above.

\end{proof}

In the following we list the sets of four generalized Bell states from $\C{4} $ which fail the necessary condition, \textbf{R} in the same fashion as example \ref{ens2}; each set listed represents an equivalence class of sets of four generalized Bell states which are local unitarily equivalent to it. These are $39$ in number. \medbreak

\begin{widetext}
\begin{center}
\begin{table}[H]
\caption{Of the $39$ sets of states which fail the necessary condition \textbf{R}, $25$ sets are listed here. All of these sets contain the states $\ket{\psi_{00}^{(4)}}$ and $\ket{\psi_{01}^{(4)}}$; the remaining states in each set are listed below.}
\begin{tabular}{|l l l l l l l |}
\hline 
 &   &  &  & & & \\
$\{ \ket{\psi_{02}^{(4)}}$, $\ket{\psi_{20}^{(4)}}\}$, &
$\{ \ket{\psi_{02}^{(4)}}$, $\ket{\psi_{22}^{(4)}}\}$, & 
$\{ \ket{\psi_{10}^{(4)}}$, $\ket{\psi_{12}^{(4)}}\}$, & 
$\{ \ket{\psi_{10}^{(4)}}$, $\ket{\psi_{21}^{(4)}}\}$, &
$\{ \ket{\psi_{10}^{(4)}}$, $\ket{\psi_{22}^{(4)}}\}$, & 
$\{ \ket{\psi_{11}^{(4)}}$, $\ket{\psi_{13}^{(4)}}\}$, &
$\{ \ket{\psi_{11}^{(4)}}$, $\ket{\psi_{20}^{(4)}}\}$, \\ 
$\{ \ket{\psi_{11}^{(4)}}$, $\ket{\psi_{23}^{(4)}}\}$, & 
$\{ \ket{\psi_{12}^{(4)}}$, $\ket{\psi_{21}^{(4)}}\}$, & 
$\{ \ket{\psi_{12}^{(4)}}$, $\ket{\psi_{22}^{(4)}}\}$, & 
$\{ \ket{\psi_{13}^{(4)}}$, $\ket{\psi_{20}^{(4)}}\}$, & 
$\{ \ket{\psi_{13}^{(4)}}$, $\ket{\psi_{23}^{(4)}}\}$, & 
$\{ \ket{\psi_{20}^{(4)}}$, $\ket{\psi_{22}^{(4)}}\}$, & 
$\{ \ket{\psi_{20}^{(4)}}$, $\ket{\psi_{23}^{(4)}}\}$, \\
$\{ \ket{\psi_{20}^{(4)}}$, $\ket{\psi_{31}^{(4)}}\}$, & 
$\{ \ket{\psi_{20}^{(4)}}$, $\ket{\psi_{33}^{(4)}}\}$, & 
$\{ \ket{\psi_{21}^{(4)}}$, $\ket{\psi_{23}^{(4)}}\}$, & 
$\{ \ket{\psi_{21}^{(4)}}$, $\ket{\psi_{30}^{(4)}}\}$, & 
$\{ \ket{\psi_{21}^{(4)}}$, $\ket{\psi_{32}^{(4)}}\}$, &
$\{ \ket{\psi_{22}^{(4)}}$, $\ket{\psi_{30}^{(4)}}\}$, &
$\{ \ket{\psi_{22}^{(4)}}$, $\ket{\psi_{32}^{(4)}}\}$, \\
$\{ \ket{\psi_{23}^{(4)}}$, $\ket{\psi_{31}^{(4)}}\}$, &
$\{ \ket{\psi_{23}^{(4)}}$, $\ket{\psi_{33}^{(4)}}\}$, &
$\{ \ket{\psi_{30}^{(4)}}$, $\ket{\psi_{32}^{(4)}}\}$, &
$\{ \ket{\psi_{31}^{(4)}}$, $\ket{\psi_{33}^{(4)}}\}$  &
&   &  \\
& & & & & & \\
\hline
\end {tabular}
\label{tab1}
\end{table}

\begin{table}[H]
\caption{The remaining $14$ sets, which fail the necessary condition \textbf{R}, contain the states $\ket{\psi_{00}^{(4)}}$ and $\ket{\psi_{02}^{(4)}}$; remaining states in each set are listed below.}
\begin{tabular}{|l l l l l l l|}
\hline
&   &  &  &  & &  \\
$\{ \ket{\psi_{10}^{(4)}}$, $\ket{\psi_{20}^{(4)}}\}$, &
$\{ \ket{\psi_{10}^{(4)}}$, $\ket{\psi_{22}^{(4)}}\}$, &
$\{ \ket{\psi_{10}^{(4)}}$, $\ket{\psi_{31}^{(4)}}\}$, &
$\{ \ket{\psi_{10}^{(4)}}$, $\ket{\psi_{32}^{(4)}}\}$, &
$\{ \ket{\psi_{10}^{(4)}}$, $\ket{\psi_{33}^{(4)}}\}$, &
$\{ \ket{\psi_{11}^{(4)}}$, $\ket{\psi_{20}^{(4)}}\}$, &
$\{ \ket{\psi_{11}^{(4)}}$, $\ket{\psi_{22}^{(4)}}\}$, \\
$\{ \ket{\psi_{11}^{(4)}}$, $\ket{\psi_{30}^{(4)}}\}$, &
$\{ \ket{\psi_{11}^{(4)}}$, $\ket{\psi_{31}^{(4)}}\}$, &
$\{ \ket{\psi_{11}^{(4)}}$, $\ket{\psi_{32}^{(4)}}\}$, &
$\{ \ket{\psi_{20}^{(4)}}$, $\ket{\psi_{30}^{(4)}}\}$, &
$\{ \ket{\psi_{20}^{(4)}}$, $\ket{\psi_{31}^{(4)}}\}$, &
$\{ \ket{\psi_{20}^{(4)}}$, $\ket{\psi_{32}^{(4)}}\}$, \\
$\{ \ket{\psi_{20}^{(4)}}$, $\ket{\psi_{33}^{(4)}}\}$  &
& & & & & \\
& & & & & & \\
\hline 
\end{tabular}
\label{tab2}
\end{table}
\end{center}
\end{widetext}
The necessary condition \textbf{R} is satisfied by all sets in all remaining $83$ equivalence classes. That doesn't mean that states in these sets \emph{should} be perfectly locally distinguishable. Therefore it comes as a surprise that states in each of the remaining $83$ equivalence classes are indeed locally distinguishable, and that too by one-way LOCC using only projective measurements. In Appendix \eqref{app} we list $83$ distinct representative sets for all of these equivalence classes and give the LOCC protocols to distinguish the states in each such set. Thus the necessary condition is also sufficient to establish the local (in)distinguishability of all sets of four generalized Bell states in $\C{4}$. 

\section{Conclusion} \label{conclusion} Based on an upper bound of locally accessible information, we formulated a very simple necessary condition for the perfect distinguishability of a set of MES by LOCC. This necessary condition is shown to be stronger than the well-known OP condition. This was illustrated in example \ref{ens2}, where it was shown that the OP condition failed to conclude whether the corresponding states were locally distinguishable or not, whereas our condition successfully determined that the states were locally indistinguishable. Also, we tested our necessary condition \textbf{R} for the local distinguishability on all sets of four generalized Bell states in $\C{4}$, and obtained those sets which fail the necessary condition \textbf{R} - these sets are given in tables \ref{tab1} and \ref{tab2}. Surprisingly, we discovered that all the remaining sets are perfectly distinguishable by one-way LOCC using only projective measurements - we explicitly obtained the LOCC protocol for perfect distinguishability for all of them (see Appendix \ref{app}). That there is no protocol which involves two-way LOCC is interestingly similar to the result in \cite{TY15}, where it was shown that two-way LOCC doesn't play any distinguished role in the perfect distinguishability of a set of four 
ququad-ququad lattice states in $\C{4}$. 

In ref \cite{Y12} a set of four ququad-ququad lattice states were shown to be indistinguishable by PPT preserving operations. It was recently shown that this is the \emph{only} such set among the sets of ququad-ququad lattice states, which isn't perfectly distinguishable by LOCC \cite{TY15}. We also tested our necessary condition \textbf{R} on the aforementioned set of states and found that they do not satisfy the necessary condition. Thus, condition \textbf{R} is also sufficien to determine the local distinguishability for ququad-ququad lattice states. In fact, our condition is more general since the condition for the local distinguishability for ququad-ququad lattice states given in \cite{TY15} is particularly specific to ququad-ququad lattice states whereas our condition \textbf{R} applies generally to any set of pairwise orthogonal MES. 

It would be interesting to know if the same behaviour can be exhibited for $d \geq 5$ as well, i.e., if the upper bound on locally accessible information is the dividing line between sets of $d$ generalized Bell states in $\C{d}$ which are perfectly distinguishable by LOCC and those sets which are not. Moreover, it is interesting to know in this context whether all those sets, which are perfectly distinguishable, are so by one-way LOCC involving projective measurement.

The condition \textbf{R} may turn out to be sufficient not just for generalized Bell states or ququad-ququad lattice states, but for some other specific classes of pairwise orthogonal MES. It would be interesting to find out the existence of such a class of states. 

\section*{Acknowledgments} Ramij Rahman acknowledges support from UGC (University Grants Commission, Govt. of India) Start-Up Grant. The authors would also like to thank Somshubhro Bandyopadhyay for useful discussions. A major part of the work was done when TS was a Ph.D. student at the Optics \& Quantum Information Group, The Institute of Mathematical Sciences, Chennai.

\begin{appendices}

\section{For $d=4$, the number of equivalence classes is $122$}
\label{122}

For any choice of $n, m, l, k \in \mathbb{Z}_4$, note that there is a unique local unitary (upto a global phase) on Alice's side $U_{n-l,m-k}$ such that

 $$\ket{\psi_{nm}^{(4)}} = \left( U_{n-l,k-m} \otimes \mathbb{1}_B \right) \ket{\psi_{lk}^{(4)}},$$
 
 where $$U_{nm} = \sum_{j \in \mathbb{Z}_4} e^{\frac{2 \pi i j n}{4}} \ketbra{j \oplus_4 m}{j}.$$

Hence, the sets $\{ \ket{\psi_{n_1m_1}^{(4)}}$, $\ket{\psi_{n_2m_2}^{(4)}}$, $\ket{\psi_{n_3m_3}^{(4)}},$ $\ket{\psi_{n_4m_4}^{(4)}} \}$ and $\{ \ket{\psi_{l_1k_1}^{(4)}},$ $\ket{\psi_{l_2k_2}^{(4)}},$ $\ket{\psi_{l_3k_3}^{(4)}},$ $\ket{\psi_{l_4k_4}^{(4)}} \}$ are unitarily equivalent if and only if there is a permutation $\sigma \in S_4$ such that $n_i - l_{\sigma(i)}$ and $m_i - k_{\sigma(i)}$ are invariant over $i \in \mathbb{Z}_4$. This condition can then be used to write a simple computer program to obtain the number of equivalence classes, in which sets of states in the same class are equivalent upto a local unitary, and sets of states belonging to different classes aren't. Using such a computer program, we found that $\binom{16}{4}$ such distinct sets can be partitioned into $122$ equivalence classes. 

\section{LOCC Protocols for Various Sets of Generalized Bell States}

\label{app}

Among the $122$ equivalence classes of sets of $4$ generalized Bell states, we here give a list of $83$ equivalence classes which satisfy the condition \textbf{R}. Each equivalence class is represented by a set of four generalized Bell states which is contained in it. While satisfying the condition \textbf{R} doesn't necessarily imply that any of these sets of generalized Bell states should be perfectly locally distinguishable, surprisingly, we find that that is indeed the case, and that too by one-way LOCC using only projective measurements. Along with each set of states we also give the one-way LOCC protocol for their perfect local distinguishability.\smallbreak
 
\begin{thm}
\label{set0} 
Any set of four generalized Bell states in $\C{4}$ of the form \small$\{ \ket{\psi_{a 0}^{(4)}},\ket{\psi_{b1}^{(4)}},\ket{\psi_{c2}^{(4)}},\ket{\psi_{d3}^{(4)}} \} $\normalsize, where $a,b,c,d \in \{0,1,2,3\}$, can be discriminated by one-way LOCC using only projective measurements.  Similarly, any set of four generalized bell states in $\C{4}$ which is of the form \small$\{ \ket{\psi_{0a}^{(4)}},\ket{\psi_{1b}^{(4)}},\ket{\psi_{2c}^{(4)}},\ket{\psi_{3d}^{(4)}} \} $\normalsize, where $a,b,c,d \in \{0,1,2,3\}$, can be discriminated by one-way LOCC using only projective measurements.
\end{thm}

\begin{proof}
Consider first the set \small$\{ \ket{\psi_{a 0}^{(4)}},\ket{\psi_{b1}^{(4)}},\ket{\psi_{c1}^{(4)}},\ket{\psi_{d1}^{(4)}} \} $\normalsize. Alice starts with a rank-one projective measurement in the following orthonormal basis: \smallbreak
 $ \ket{u_0} = \;\ket{0}$, $ \ket{u_1} = \;\ket{1}$, $ \ket{u_2} = \;\ket{2}$, $ \ket{u_3} = \;\ket{3}$. 
\medbreak
For the $k$-th outcome of Alice's measurement, the post measurement set of joint states will be of the following form: 
\begin{align*}
&\{ \ket{\psi_{a 0}^{(4)}},\ket{\psi_{b1}^{(4)}},\ket{\psi_{c2}^{(4)}},\ket{\psi_{d3}^{(4)}} \} \notag \\ \longrightarrow & \{ \ket{k} \ket{k},\ket{k} \ket{k\oplus_4 1},\ket{k} \ket{k\oplus_4 2},\ket{k} \ket{k\oplus_4 3} \}. 
\end{align*}

Thus, once Alice tells Bob her measurement's outcome, he needs to perform measurement in the $\{ \ket{j}\}_{j = 0}^{4}$ basis to perfectly distinguish between the states in the set.  \medbreak

Next, consider a set of states of the form $\{ \ket{\psi_{0a}^{(4)}}$, $\ket{\psi_{1b}^{(4)}}$, $\ket{\psi_{2c}^{(4)}}$, $\ket{\psi_{3d}^{(4)}} \} $.
Alice starts by performing a rank-one projective measurement corresponding to the following orthonormal basis: \medbreak
\small
$ \ket{u_0} = \;\frac{1}{2}\sum_{j=0}^{3} \ket{j}$,
$ \ket{u_1} = \;\frac{1}{2}\sum_{j=0}^{3} e^{\frac{ij \pi }{2}} \ket{j}$,
$ \ket{u_2} = \;\frac{1}{2}\sum_{j=0}^{3} (-1)^j \ket{j}$,
$ \ket{u_3} = \;\frac{1}{2}\sum_{j=0}^{3} e^{\frac{3ij \pi }{2}} \ket{j}$. \medbreak
\normalsize
For the $k$-th outcome of Alice's measurement, the post measurement set of joint states will be of the following form: 

\begin{align*}
& \{ \ket{\psi_{0a}^{(4)}},\ket{\psi_{1b}^{(4)}},\ket{\psi_{2c}^{(4)}},\ket{\psi_{3d}^{(4)}} \}\notag  \\\longrightarrow & \{ \ket{u_k}\ket{v_k},\ket{u_k}\ket{v_{k\oplus_4 1}},\ket{u_k}\ket{v_{k\oplus_4 2}},\ket{u_k}\ket{v_{k\oplus_4 3}}\},
\end{align*}

where \medbreak
$ \ket{v_0} = \;\frac{1}{2} \sum_{j=0}^{3} e^{\frac{(4 \ominus_4 k)ij \pi}{2}} \ket{j}$,  

$ \ket{v_1} = \;\frac{1}{2}\sum_{j=0}^{3} e^{\frac{(1 \ominus_4 k)ij \pi}{2}} \ket{j}$, 

$ \ket{v_2} = \;\frac{1}{2}\sum_{j=0}^{3} e^{\frac{(2 \ominus_4 k)ij \pi}{2}} \ket{j}$, 

$ \ket{v_3} = \;\frac{1}{2}\sum_{j=0}^{3} e^{\frac{(3 \ominus_4 k)ij \pi}{2}} \ket{j}$. 
 
Thus, once Alice tells Bob her measurement's outcome, he needs to perform measurement in the $\{ \ket{v_j}\}_{j=0}^{3}$ basis to perfectly distinguish between the states in the set. \end{proof}

\begin{thm} \label{set1} States in each set in the following two tables are perfectly distinguishable by one-way LOCC using only projective measurement: \smallbreak
\small
\begin{tabular}{| l l l|}
\hline 
\multicolumn{3}{|c|}{ Each of the following $10$ sets contains the states}\\
\multicolumn{3}{|c|}{$\ket{\psi_{00}^{(4)}}$ and $\ket{\psi_{01}^{(4)}}$; remaining states are listed below.}\\
\hline 
 &  &     \\
 $\{   \ket{\psi_{02}^{(4)}}$, $\ket{\psi_{10}^{(4)}}\}$, & 
 $\{   \ket{\psi_{02}^{(4)}}$, $\ket{\psi_{21}^{(4)}}\}$, & 
 $\{   \ket{\psi_{02}^{(4)}}$, $\ket{\psi_{32}^{(4)}}\}$, \\
 $\{   \ket{\psi_{10}^{(4)}}$, $\ket{\psi_{13}^{(4)}}\}$, &
 $\{   \ket{\psi_{10}^{(4)}}$, $\ket{\psi_{20}^{(4)}}\}$, &
 $\{   \ket{\psi_{10}^{(4)}}$, $\ket{\psi_{31}^{(4)}}\}$, \\ 
 $\{   \ket{\psi_{13}^{(4)}}$, $\ket{\psi_{21}^{(4)}}\}$, &
 $\{   \ket{\psi_{20}^{(4)}}$, $\ket{\psi_{21}^{(4)}}\}$, &
 $\{   \ket{\psi_{20}^{(4)}}$, $\ket{\psi_{32}^{(4)}}\}$, \\
 $\{   \ket{\psi_{21}^{(4)}}$, $\ket{\psi_{31}^{(4)}}\}$. &
 & \\
  & &    \\
\hline 
\end{tabular} \bigbreak
 
\begin{tabular}{|l l l  |}
\hline
\multicolumn{3}{|c|}{ Each of the following $4$ sets contains the states}\\
\multicolumn{3}{|c|}{$\ket{\psi_{00}^{(4)}}$ and $\ket{\psi_{02}^{(4)}}$; remaining states are listed below.}\\
\hline 
 &   &    \\
 $\{  \ket{\psi_{10}^{(4)}}$, $\ket{\psi_{12}^{(4)}}\}$, &
 $\{  \ket{\psi_{10}^{(4)}}$, $\ket{\psi_{23}^{(4)}}\}$, &
 $\{  \ket{\psi_{10}^{(4)}}$, $\ket{\psi_{30}^{(4)}}\}$, \\
 $\{  \ket{\psi_{21}^{(4)}}$, $\ket{\psi_{30}^{(4)}}\}$. &
 & \\  &   &  \\
\hline 
\end{tabular} 
 
\end{thm}

\begin{proof}
Alice performs a rank-one projective measurement in the following orthonormal basis: \smallbreak

$ \ket{u_1} = \frac{1}{2} \left(- e^{i \frac{\pi}{4}}\ket{0}+\ket{1}+ e^{i \frac{\pi}{4}}\ket{2}+\ket{3}\right)$, 

$ \ket{u_2} = \frac{1}{2}\left( e^{i \frac{\pi}{4}}\ket{0}+\ket{1}- e^{i \frac{\pi}{4}}\ket{2}+\ket{3}\right)$,
 
$ \ket{u_3} = \frac{1}{2}\left( e^{i \frac{3\pi}{4}}\ket{0}-\ket{1}+ e^{i \frac{3\pi}{4}}\ket{2}+\ket{3}\right)$,
 
$ \ket{u_4} = -\frac{1}{2} \left( e^{i \frac{3\pi}{4}}\ket{0}-\ket{1}- e^{i \frac{3\pi}{4}}\ket{2}+\ket{3}\right)$. \smallbreak

For each measurement outcome, Bob's PMRS are orthogonal and hence, distinguishable.

\end{proof}

\begin{thm} \label{set2} The following sets are distinguishable by one-way LOCC using only projective measurement: \medbreak
\small
\begin{tabular}{|c c   |}
\hline
\multicolumn{2}{|c|}{ Each of the following  $8$ sets contains the state}\\
\multicolumn{2}{|c|}{$\ket{\psi_{00}^{(4)}}$, the rest are given in the table.}\\
\hline
&  \\
 $\{ \ket{\psi_{01}^{(4)}}$, $\ket{\psi_{02}^{(4)}}$, $\ket{\psi_{11}^{(4)}}\}$, & 
 $\{ \ket{\psi_{01}^{(4)}}$, $\ket{\psi_{02}^{(4)}}$, $\ket{\psi_{31}^{(4)}}\}$, \\
 $\{ \ket{\psi_{01}^{(4)}}$, $\ket{\psi_{10}^{(4)}}$, $\ket{\psi_{23}^{(4)}}\}$, & 
 $\{ \ket{\psi_{01}^{(4)}}$, $\ket{\psi_{11}^{(4)}}$, $\ket{\psi_{22}^{(4)}}\}$, \\ 
 $\{ \ket{\psi_{01}^{(4)}}$, $\ket{\psi_{22}^{(4)}}$, $\ket{\psi_{31}^{(4)}}\}$, & 
 $\{ \ket{\psi_{01}^{(4)}}$, $\ket{\psi_{23}^{(4)}}$, $\ket{\psi_{30}^{(4)}}\}$, \\
 $\{ \ket{\psi_{02}^{(4)}}$, $\ket{\psi_{11}^{(4)}}$, $\ket{\psi_{21}^{(4)}}\}$, & 
 $\{ \ket{\psi_{02}^{(4)}}$, $\ket{\psi_{21}^{(4)}}$, $\ket{\psi_{31}^{(4)}}\}$  \\
&  \\
 \hline
\end{tabular}
\end{thm}

\begin{proof}
Alice performs a rank-one projective measurement in the following orthonormal basis: \smallbreak
$ \ket{u_1} =  \frac{1}{\sqrt{2}} \left(-i\ket{0}+\ket{2}\right)$,

$ \ket{u_2} =  \frac{1}{\sqrt{2}} \left(i\ket{0}+\ket{2}\right)$,  

$ \ket{u_3} =  \frac{1}{\sqrt{2}} \left(-i\ket{1}+\ket{3}\right)$, 

$ \ket{u_4} =  \frac{1}{\sqrt{2}} \left(i\ket{1}+\ket{3}\right)$.  \smallbreak
For each measurement outcome, Bob's PMRS are found to be orthogonal and hence can be distinguished.
\end{proof}

\begin{thm} \label{set3} The following sets are distinguishable by one-way LOCC using only projective measurement: \medbreak
\small
\begin{tabular}{|c c  |}
\hline
\multicolumn{2}{|c|}{ Each of the following  $8$ sets contains the state}\\
\multicolumn{2}{|c|}{$\ket{\psi_{00}^{(4)}}$, the rest are given in the table.}\\
\hline
&   \\
$\{  \ket{\psi_{01}^{(4)}}$, $\ket{\psi_{02}^{(4)}}$, $\ket{\psi_{12}^{(4)}}\}$, &  
$\{  \ket{\psi_{01}^{(4)}}$, $\ket{\psi_{02}^{(4)}}$, $\ket{\psi_{30}^{(4)}}\}$, \\ 
$\{  \ket{\psi_{01}^{(4)}}$, $\ket{\psi_{11}^{(4)}}$, $\ket{\psi_{21}^{(4)}}\}$, & 
$\{  \ket{\psi_{01}^{(4)}}$, $\ket{\psi_{12}^{(4)}}$, $\ket{\psi_{20}^{(4)}}\}$, \\ 
$\{  \ket{\psi_{01}^{(4)}}$, $\ket{\psi_{20}^{(4)}}$, $\ket{\psi_{30}^{(4)}}\}$, & 
$\{  \ket{\psi_{01}^{(4)}}$, $\ket{\psi_{21}^{(4)}}$, $\ket{\psi_{33}^{(4)}}\}$, \\ 
$\{  \ket{\psi_{02}^{(4)}}$, $\ket{\psi_{10}^{(4)}}$, $\ket{\psi_{21}^{(4)}}\}$, & 
$\{  \ket{\psi_{02}^{(4)}}$, $\ket{\psi_{21}^{(4)}}$, $\ket{\psi_{32}^{(4)}}\}$. \\
&  \\ 
\hline
\end{tabular}
\end{thm}

\begin{proof}

Alice performs a rank-one projective measurement in the following orthonormal basis: \smallbreak
\small
$ \ket{u_1} = \frac{1}{2} \left( -e^{i \frac{3\pi}{4}}\ket{0}+\ket{1}+ e^{i \frac{3\pi}{4}}\ket{2}+\ket{3} \right)$,

$ \ket{u_2} = \frac{1}{2} \left( e^{i \frac{3\pi}{4}}\ket{0}+\ket{1}-e^{i \frac{3\pi}{4}}\ket{2}+\ket{3}\right)$,
 
$ \ket{u_3} = \frac{1}{2} \left( e^{i \frac{\pi}{4}}\ket{0}-\ket{1}+e^{i \frac{\pi}{4}}\ket{2}+\ket{3} \right)$,
$ \ket{u_4} = \frac{1}{2} \left( -e^{i \frac{\pi}{4}}\ket{0}-\ket{1}- e^{i \frac{\pi}{4}}\ket{2}+\ket{3} \right)$. \medbreak
 
 For each measurement outcome, Bob's PMRS are found to be orthogonal and hence can be distinguished.

\end{proof}

\begin{thm} \label{set4} The following sets are distinguishable by one-way LOCC using only projective measurement: \medbreak
\small
\begin{tabular}{|c c  |}
\hline
\multicolumn{2}{|c|}{ Each of the following  $10$ sets contains the state}\\
\multicolumn{2}{|c|}{$\ket{\psi_{00}^{(4)}}$ the rest are given in the table.}\\
\hline
& \\
$\{  \ket{\psi_{01}^{(4)}}$, $\ket{\psi_{10}^{(4)}}$, $\ket{\psi_{11}^{(4)}}\}$, &
$\{  \ket{\psi_{01}^{(4)}}$, $\ket{\psi_{10}^{(4)}}$, $\ket{\psi_{30}^{(4)}}\}$, \\ 
$\{  \ket{\psi_{01}^{(4)}}$, $\ket{\psi_{10}^{(4)}}$, $\ket{\psi_{32}^{(4)}}\}$, & 
$\{  \ket{\psi_{01}^{(4)}}$, $\ket{\psi_{11}^{(4)}}$, $\ket{\psi_{12}^{(4)}}\}$, \\ 
$\{  \ket{\psi_{01}^{(4)}}$, $\ket{\psi_{11}^{(4)}}$, $\ket{\psi_{31}^{(4)}}\}$, & 
$\{  \ket{\psi_{01}^{(4)}}$, $\ket{\psi_{11}^{(4)}}$, $\ket{\psi_{33}^{(4)}}\}$, \\
$\{  \ket{\psi_{01}^{(4)}}$, $\ket{\psi_{12}^{(4)}}$, $\ket{\psi_{30}^{(4)}}\}$, & 
$\{  \ket{\psi_{01}^{(4)}}$, $\ket{\psi_{12}^{(4)}}$, $\ket{\psi_{32}^{(4)}}\}$, \\ 
$\{  \ket{\psi_{01}^{(4)}}$, $\ket{\psi_{13}^{(4)}}$, $\ket{\psi_{31}^{(4)}}\}$, & 
$\{  \ket{\psi_{01}^{(4)}}$, $\ket{\psi_{13}^{(4)}}$, $\ket{\psi_{33}^{(4)}}\}$. \\
& \\ 
\hline
\end{tabular}
\end{thm}

\begin{proof}

Alice performs a rank-one projective measurement in the following orthonormal basis: \smallbreak
$ \ket{u_1} = \;\frac{1}{2} \left( \ket{0}-\ket{1}-\ket{2}+\ket{3} \right)$,

$ \ket{u_2} = \;\frac{1}{2} \left(-\ket{0}-\ket{1}+\ket{2}+\ket{3} \right)$, 

$ \ket{u_3} = \;\frac{1}{2} \left(-\ket{0}+\ket{1}-\ket{2}+\ket{3} \right)$,

$ \ket{u_4} = \;\frac{i}{2} \left( \ket{0}+\ket{1}+\ket{2}+\ket{3} \right)$. \smallbreak
 For each measurement outcome, Bob's PMRS are found to be orthogonal and hence can be distinguished.

\end{proof}

\begin{thm} \label{set5} The following sets are distinguishable by one-way LOCC using only projective measurement: \smallbreak
\small
\begin{tabular}{|c c |}
\hline
\multicolumn{2}{|c|}{ Each of the following  $5$ sets contains the state}\\
\multicolumn{2}{|c|}{$\ket{\psi_{00}^{(4)}}$ the rest are given in the table.}\\
\hline
&  \\
$\{  \ket{\psi_{01}^{(4)}}$, $\ket{\psi_{10}^{(4)}}$, $\ket{\psi_{33}^{(4)}}\}$, &
$\{  \ket{\psi_{01}^{(4)}}$, $\ket{\psi_{11}^{(4)}}$, $\ket{\psi_{32}^{(4)}}\}$, \\
$\{  \ket{\psi_{01}^{(4)}}$, $\ket{\psi_{12}^{(4)}}$, $\ket{\psi_{31}^{(4)}}\}$, & 
$\{  \ket{\psi_{01}^{(4)}}$, $\ket{\psi_{13}^{(4)}}$, $\ket{\psi_{30}^{(4)}}\}$, \\ 
$\{ \ket{\psi_{01}^{(4)}}$, $\ket{\psi_{11}^{(4)}}$, $\ket{\psi_{30}^{(4)}}\}$.  & \\
&  \\
\hline
\end{tabular}
\end{thm}

\begin{proof}

Alice performs a rank-one projective measurement in the following orthonormal basis: \smallbreak
\small
$ \ket{u_1} = \; \frac{1}{\sqrt{2}}\left(-\ket{0}+\ket{2} \right)$,

$ \ket{u_2} = \; \frac{1}{\sqrt{2}}\left(\ket{1}+\ket{3} \right)$,

$ \ket{u_3} = \; \frac{1}{\sqrt{2}}\left(-\ket{1}+\ket{3} \right)$,

$ \ket{u_4} = \; \frac{1}{\sqrt{2}}\left( \ket{0}+\ket{2} \right)$. \smallbreak
For each measurement outcome, Bob's PMRS are found to be orthogonal and hence can be distinguished.

\end{proof}

\begin{thm} \label{set6} The following set is distinguishable by one-way LOCC using only projective measurement: \smallbreak

 $\{ \ket{\psi_{00}^{(4)}}$,  $\ket{\psi_{02}^{(4)}}$, $\ket{\psi_{20}^{(4)}}$, $\ket{\psi_{22}^{(4)}}\}$.  

\end{thm}

\begin{proof}

Alice performs a rank-one projective measurement in the following orthonormal basis: \smallbreak
$ \ket{u_1} = \;\frac{1}{2}\left(-i\ket{0}-i\ket{1}+\ket{2}+\ket{3} \right)$,

$ \ket{u_2} = \;\frac{1}{2}\left(i\ket{0}+i\ket{1}+\ket{2}+\ket{3}  \right)$,

$ \ket{u_3} = \;\frac{1}{2}\left(i\ket{0}-i\ket{1}-\ket{2}+\ket{3}  \right)$,

$ \ket{u_4} = \;\frac{1}{2}\left(-i\ket{0}+i\ket{1}-\ket{2}+\ket{3}  \right)$. \smallbreak
For each measurement outcome, Bob's PMRS are found to be orthogonal and hence can be distinguished.

\end{proof} 
\end{appendices}
\bibliographystyle{plain}
\bibliography{mybib}

\end{document}